\documentclass{article}
\usepackage{amsmath,amssymb,amsthm,verbatim}
\usepackage[margin=1in]{geometry}

\newcommand{\sgn}{\textrm{sgn}}
\newcommand{\R}{\mathbb{R}}
\newcommand{\E}{\mathbb{E}}
\newcommand{\pr}{\textrm{Pr}}

\newcommand{\ns}{\mathbb{NS}}
\newcommand{\as}{\mathbb{AS}}

\newtheorem{thm}{Theorem}
\newtheorem{prop}[thm]{Proposition}
\newtheorem{cor}[thm]{Corollary}
\newtheorem{lem}[thm]{Lemma}
\newtheorem*{claim}{Claim}
\newtheorem*{defn}{Definition}
\newtheorem*{rmk}{Remark}

\title{The Average Sensitivity of an Intersection of Half Spaces}
\author{Daniel M. Kane}

\begin{document}
\maketitle
\section{Introduction}\label{IntroSec}

One of the most important measures of the complexity of a Boolean function $f:\R^n\rightarrow \{\pm 1\}$ is that of its average sensitivity, namely
$$
\as(f) := \E_{x\sim_u\{\pm 1\}^n}\left[\#\{i:f(x)\neq f(x^i)\} \right]
$$
where $x^i$ above is $x$ with the $i^{th}$ coordinate flipped. The average sensitivity and related measures of noise sensitivity of a Boolean function have found several applications, perhaps most notably to the area of machine learning (see for example \cite{learning}). It has thus become important to understand how large the average sensitivity of functions in various classes can be.

Of particular interest is the study of the sensitivity of certain classes of algebraically defined functions. Gotsman and Linial (\cite{gl}) first studied the sensitivity of polynomial threshold functions (i.e. functions of the form $f(x)=\sgn(p(x))$ for $p$ a polynomial of bounded degree). They conjectured exact upper bounds on the sensitivity of polynomial threshold functions of limited degree, but were unable to prove them except in the case of linear threshold functions (when $p$ is required to be degree $1$). Since then, significant progress has been made towards proving this Conjecture. The first non-trivial bounds for large degree were proven in \cite{sensitivity} by Diakonikolas et. al. in 2010. Since then, progress has been rapid. In \cite{gsurf}, the Gaussian analogue of the Gotsman-Linial Conjecture was proved, and in \cite{glexp} the correct bound on average sensitivity was proved to within a polylogarithmic factor.

Another potential generalization of the degree-$1$ case of the Gotsman-Linial Conjecture (which bounds the sensitivity of the indicator function of a halfspace) would be to consider the sensitivity of the indictor function of the intersection of a bounded number of halfspaces. The Gaussian analogue of this question has already been studied. In particular, Nazarov has shown (see \cite{kn:learning}) that the Gaussian surface area of an intersection of $k$ halfspaces is at most $O(\sqrt{\log k})$. This suggests that the average sensitivity of such a function should be bounded by $O(\sqrt{n\log k})$. Although this bound has been believed for some time, attempts to prove it have been unsuccessful. Perhaps the closest attempt thus far was by Harsha, Klivans and Meka who show in \cite{regular} that an intersection of $k$ sufficiently regular halfspaces has noise sensitivity with parameter $\epsilon$ at most $\log(k)^{O(1)}\epsilon^{1/6}$. In this paper, we prove that the bound of $O(\sqrt{n\log(k)})$ is in face correct. In particular, we prove the following Theorem:
\begin{thm}\label{mainThm}
Let $f$ be the indicator function of an intersection of $k$ half spaces in $n$ variables, then
$$
\as(f) = O(\sqrt{n\log(k)}).
$$
\end{thm}
It should also be noted that Nazarov's bound follows as a Corollary of Theorem \ref{mainThm}, by replacing Gaussian random variables with averages of Bernoulli random variables. It is also not hard to show that this bound is tight up to constants. In particular:
\begin{thm}\label{lowerBoundThm}
If $k\leq 2^n$, there exists a function $f$ in $n$ variables given by the intersection of at most $k$ half spaces so that
$$
\as(f)=\Omega(\sqrt{n\log(k)}).
$$
\end{thm}

Our proof of Theorem \ref{mainThm} actually uses very little information about halfspaces. In particular, we use only the fact that linear threshold functions are monotonic in the following sense:
\begin{defn}
We say that a function $f:\{\pm 1\}^n\rightarrow \R$ is \emph{unate} if for all $i$, $f$ is either increasing with respect to the $i^{th}$ coordinate or decreasing with respect to the $i^{th}$ coordinate.
\end{defn}

We prove Theorem \ref{mainThm} by means of the following much more general statement:
\begin{prop}\label{mainProp}
Let $f_1,\ldots,f_k:\{\pm 1\}^n\rightarrow \{0,1\}$, be unate functions and let $F:\{\pm 1\}^n\rightarrow \{0,1\}$ be defined as $F(x)=\bigvee_{i=1}^k f_i(x)$. Then
$$
\as(F) =  O(\sqrt{n\log(k)}).
$$
\end{prop}

The application of Theorem \ref{mainThm} to machine learning is via a slightly different notion of noise sensitivity than that of the average sensitivity. In particular, we define the \emph{noise sensitivity} as follows
\begin{defn}
Let $f:\{\pm 1\}^n\rightarrow \{0,1\}$ be a Boolean function. For a parameter $\epsilon\in (0,1)$ we define the noise sensitivity of $f$ with parameter $\epsilon$ to be
$$
\ns_\epsilon(f) := \pr(f(x)\neq f(y))
$$
where $x$ and $y$ are Bernoulli random variables where $y$ is obtained from $x$ by randomly and independently flipping the sign of each coordinate with probability $\epsilon$.
\end{defn}

Using this notation, we have that
\begin{cor}\label{nsBoundCor}
If $f:\{\pm 1\}^n\rightarrow\{0,1\}$ is the indicator function of the intersection of $k$ halfspaces, and $\epsilon\in(0,1)$ then
$$
\ns_\epsilon(f) = O(\sqrt{\epsilon \log(k)}).
$$
\end{cor}
\begin{rmk}
This is false in general for intersections of unate functions, since if $f$ is the tribes function on $n$ variables (which is unate) then $\ns_\epsilon(f)=\Omega(1)$ so long as $\epsilon=\Omega(\log^{-1}(n))$.

It should also be noted that this problem has been considered in the case when $k=1$ in \cite{ns}, and proven (again for $k=1$) in \cite{yuval}.
\end{rmk}

Finally, using the $L^1$ polynomial regression algorithm of \cite{learning}, we obtain the following:
\begin{cor}\label{LearningCor}
The concept class of intersections of $k$ halfspaces with respect to the uniform distribution on $\{\pm 1\}^n$ is agnostically learnable with error $\mathbf{opt}+\epsilon$ in time $n^{O(\log(k)\epsilon^{-2})}$.
\end{cor}
\begin{rmk}
The problem of learning intersections of halfspaces was considered in \cite{learninghalfspaces}, where they achieved a bound of $n^{O(k^2/\epsilon^2)}$, which is substantially improved by the above.
\end{rmk}

\section{Proofs of the Sensitivity Bounds}\label{mainSec}

The proof of Proposition \ref{mainProp} follows by a fairly natural generalization of one of the standard proofs for the case of a single unate function. In particular, if $f:\{\pm 1\}^n\rightarrow\{0,1\}$ is unate, we may assume without loss of generality that $f$ is increasing in each coordinate. In such a case, it is easy to show that $$\as(f)=\E\left[ f(x)\sum_{i=1}^n x_i \right] \leq \E\left[\max\left(0,\sum_{i=1}^n x_i\right) \right]=O(\sqrt{n}).$$ In fact, this technique can be extended to prove bounds on the sensitivity of unate functions with given expectation. In particular, Lemma \ref{heightLem} below provides an appropriate bound. Our proof of Proposition \ref{mainProp} turns out to be a relatively straightforward generalization of this technique. In particular, we show that by adding the $f_i$ one at a time, the change in sensitivity is bounded by a similar function of the change in total expectation.

\begin{lem}\label{heightLem}
Let $S:\{\pm 1\}^n\rightarrow \{0,1\}$ and let $p=\E[S(x)]$, then
$$
\E\left[S(x)\sum_{i=1}^n x_i \right] = O(p\sqrt{n\log(1/p)}).
$$
\end{lem}
\begin{proof}
Note that:
\begin{align*}
\E\left[S(x)\sum_{i=1}^n x_i \right] & \leq \int_0^\infty \pr\left(S(x)\sum_{i=1}^n x_i > y\right) dy\\
& \leq \int_0^\infty \min\left( p , \pr\left(\sum_{i=1}^n x_i > y\right) \right) dy\\
& \leq \int_0^\infty \min\left( p , \exp\left(-\Omega\left( y^2/n\right)\right) \right) dy\\
& \leq O\left(\int_0^\infty \min\left( p , \exp\left(- z^2/n \right) dz \right)\right)\\
& \leq O\left( \int_0^{\sqrt{n\log(1/p)}}pdz + \int_{\sqrt{n\log(1/p)}}^\infty \exp(-z^2/n)dz\right)\\
& \leq O(p\sqrt{n\log(1/p)}).
\end{align*}
\end{proof}

We now prove Proposition \ref{mainProp}.
\begin{proof}
Let $F_m=\bigvee_{i=1}^m f_i(x)$. Let $S_m(x) = F_m(x) - F_{m-1}(x)$. Let $p_m = \E[S_m(x)]$. Our main goal will be to show that $\as(F_m)\leq \as(F_{m-1})+O(p_m\sqrt{n\log(p_m)})$, from which our result follows easily.

Consider $\as(F_m)-\as(F_{m-1})$. We assume without loss of generality that $f_m$ is increasing in every coordinate.
\begin{align*}
\as(F_m)-\as(F_{m-1}) & = \sum_{i=1}^n \E\left[\left|F_m(x)-F_m(x^i) \right| - \left|F_{m-1}(x)-F_{m-1}(x^i) \right| \right],
\end{align*}
where $x^i$ denotes $x$ with the $i^{th}$ coordinate flipped. We make the following claim:
\begin{claim}
For each $x,i$,
\begin{align}\label{sgnEqn}
\left|F_m(x)-F_m(x^i) \right| & - \left|F_{m-1}(x)-F_{m-1}(x^i) \right| \notag \\ &  \leq  x_i\left(\left(F_m(x)-F_m(x^i) \right) - \left(F_{m-1}(x)-F_{m-1}(x^i) \right)\right).
\end{align}
\end{claim}
\begin{proof}
Our proof is based on considering two different cases.

\noindent \textbf{Case 1:} $f_m(x)=f_m(x^i)=0$

In this case, $F_m(x)=F_{m-1}(x)$ and $F_m(x^i)=F_{m-1}(x^i)$, and thus both sides of Equation \eqref{sgnEqn} are 0.

\noindent \textbf{Case 2:} $f_m(x)=1$ or $f_m(x^i)=1$

Note that replacing $x$ by $x^i$ leaves both sizes of Equation \eqref{sgnEqn} the same. We may therefore assume without loss of generality that $x_i=1$. Since $f_m$ is increasing with respect to the $i^{th}$ coordinate, $f_m(x)\geq f_m(x^i)$. Since at least one of them is 1, $f_m(x)=1$. Therefore, $F_m(x)=1$. Therefore, since
$$
x_i\left(F_m(x)-F_m(x^i)\right) \geq \left|F_m(x)-F_m(x^i) \right|,
$$
and
$$
- x_i\left(F_{m-1}(x)-F_{m-1}(x^i) \right) \geq - \left|F_{m-1}(x)-F_{m-1}(x^i) \right|,
$$
Equation \eqref{sgnEqn} follows.
\end{proof}
By the claim we have that
\begin{align*}
\as(F_m)-\as(F_{m-1}) & \leq \sum_{i=1}^n \E\left[x_i\left(\left(F_m(x)-F_m(x^i) \right) - \left(F_{m-1}(x)-F_{m-1}(x^i) \right)\right)  \right]\\
& = \sum_{i=1}^n \E\left[x_i\left(S_m(x)-S_m(x^i) \right) \right]\\
& = \sum_{i=1}^n\E\left[x_iS_m(x) \right] - \sum_{i=1}^n \E\left[ (-y_i)S_m(y)\right]\\
& = 2\E\left[S_m(x)\sum_{i=1}^n x_i \right]\\
& = O(p_m\sqrt{n\log(1/p_m)}).
\end{align*}
Where the on the third line above, we are letting $y=x^i$, and the last line is by Lemma \ref{heightLem}.

Hence, we have that
\begin{align*}
\as(F) & = \sum_{m=1}^k \as(F_m)-\as(F_{m-1})\\
& = O\left(\sqrt{n}\sum_{m=1}^k p_m\sqrt{\log(1/p_m)} \right).
\end{align*}
Let $P=\E[F(x)] = \sum_{m=1}^k p_m$. By concavity of the function $x\sqrt{\log(1/x)}$ for $x\in (0,1)$, we have that
$$
\as(F) = O(\sqrt{n}P\sqrt{\log(k/P)}) = O(\sqrt{n\log(k)}).
$$
This completes our proof.
\end{proof}

Theorem \ref{mainThm} follows from Proposition \ref{mainProp} upon noting that $1-f$ is a disjunction of $k$ linear threshold functions, each of which is unate. Our proof of Theorem \ref{mainThm} shows that the bound can be tight only if a large number of the halfspaces cut off an incremental volume of roughly $1/k$. It turns out that this bound can be achieved when we take a random collection of halfspaces with such volumes. Before we begin to prove Theorem \ref{lowerBoundThm}, we need the following Lemma:

\begin{lem}\label{singleLTFLem}
For an integer $n$ and $1/2>\epsilon>2^{-n}$ there exists a linear threshold function $f:\{\pm 1\}^n\rightarrow \{0,1\}$ so that
$$
\E_{x}[f(x)] \geq \epsilon,
$$
and
$$
\as(f) = \Omega(\E_x[f(x)] \sqrt{n\log(1/\epsilon)}).
$$
\end{lem}
\begin{proof}
This is easily seen to be the case if we let $f(x)$ be the indicator function of $\sum_{i=1}^n x_i > t$ for $t$ as large as possible so that this event takes place with probability at least $\epsilon.$
\end{proof}

\begin{proof}[Proof of Theorem \ref{lowerBoundThm}]
We note that it suffices to show that there is such as $f$ given as the indicator function of a union of at most $k$ half-spaces, as $1-f$ will have the same average sensitivity and will be the indicator function of an intersection. Let $\epsilon=1/k$, and let $f$ be the function given to us in Lemma \ref{singleLTFLem}. We note that if $\E[f(x)] > 1/4$, then $f$ is sufficient and we are done. Otherwise let $m=\lfloor 1/(4\E[f(x)]) \rfloor\leq k$. For $s\in\{\pm 1\}^n$ let $f_s(x)=f(s_1x_1,\ldots,s_nx_n).$ We note for each $s$ that $f_s(x)$ is a linear threshold function with $\E[f_s(x)] = \E[f(x)]$ and $\as(f_s)=\as(f)$.

Let
$$
F(x) = \bigvee_{i=1}^m f_{s_i}(x)
$$
for $s_i$ independent random elements of $\{\pm 1\}^n$. We note that $F(x)$ is always the indicator of a union of at most $k$ half-spaces, but we also claim that
$$
\E_{s_i}[\as(F)] = \Omega\left(\sqrt{n\log(k)}\right).
$$
This would imply our result for appropriately chosen values of the $s_i$.

We note that $\as(F)$ is $2^{1-n}$ times the number of pairs of adjacent elements $x,y$ of the hypercube so that $F(x)=1,F(y)=0$. This in turn is at least $2^{1-n}$ times the sum over $1\leq i \leq m$ of the number of pairs of adjacent elements of the hypercube $x,y$ so that $f_{s_i}(x)=1,f_{s_i}(y)=0$ and so that $f_{s_j}(x)=f_{s_j}(y)=0$ for all $j\neq i$.

On the other hand, for each $i$, $2^{1-n}$ times the number of pairs of adjacent elements $x,y$ so that $f_{s_i}(x)=1,f_{s_i}(y)=0$ is $$\as(f_{s_i})=\as(f)=\Omega(\E[f(x)]\sqrt{n\log(k)}) = \Omega(m^{-1}\sqrt{n\log(k)}).$$ For each of these pairs, we consider the probability over the choice of $s_j$ that $f_{s_j}(x)=1$ or $f_{s_j}(y)=1$ for some $j\neq i$. We note that for each fixed $x$ and $j$ that
$$
\pr_{s_j}(f_{s_j}(x)=1)=\E_{s_j}[f_{s_j}(x)]=\E_{s_j}[f_x(s_j)] = \E_z[f(z)] \leq \frac{1}{4m}.
$$
Thus, by the union bound, the probability that either $f_{s_j}(x)=1$ or $f_{s_j}(y)=1$ for some $j\neq i$ is at most $1/2$. Therefore, the expected number of adjacent pairs $x,y$ with $f_{s_i}(x)=1$, $f_{s_i}(y)=0$ and $f_{s_j}(x)=f_{s_j}(y)=0$ for all $j\neq i$ is at least $\as(f_{s_j})/2$. Therefore,
$$
\E_{s_i}[\as(F)] \geq \sum_{i=1}^m \as(f)/2 = m  \Omega(m^{-1}\sqrt{n\log(k)}) =  \Omega(\sqrt{n\log(k)}),
$$
as desired. This completes our proof.

\end{proof}

\section{Learning Theory Application}\label{AppSec}

The proofs of Corollaries \ref{nsBoundCor} and \ref{LearningCor} are by what are now fairly standard techniques, but are included here for completeness. The proof of Corollary \ref{nsBoundCor} is by a technique of Diakonikolas et. al. in \cite{sense3} for bounding the noise sensitivity in terms of the average sensitivity.
\begin{proof}[Proof of Corollary \ref{nsBoundCor}]
As the noise sensitivity is an increasing function of $\epsilon$, we may round $\epsilon$ down to $1/\lceil \epsilon^{-1} \rceil$, and thus it suffices to consider $\epsilon=1/m$ for some integer $m$. We note that the pair of random variables $x$, $y$ used to define the noise sensitivity with parameter $\epsilon$ can be generated in the following way:
\begin{enumerate}
\item Randomly divide the $n$ coordinates into $m$ bins.
\item Randomly assign each coordinate a value in $\{\pm 1\}$ to obtain $z$.
\item For each bin randomly pick $b_i\in \{\pm 1\}$. Obtain $x$ from $z$ by multiplying all coordinates in the $i^{th}$ bin by $b_i$ for each $i$.
\item Obtain $y$ from $x$ by flipping the sign of all coordinates in a randomly chosen bin.
\end{enumerate}
We note that this produces the same distribution on $x$ and $y$ since $x$ is clearly a uniform element of $\{\pm 1\}^n$ and the $i^{th}$ coordinate of $y$ differs from the corresponding coordinate of $x$ if and only if $i$ lies in the bin selected in step 4. This happens independently and with probability $1/m$ for each coordinate.

Next let $f$ be the indicator function of an intersection of at most $k$ halfspaces. Note that after the bins and $z$ are picked in steps 1 and 2 above that $f(x)$ is given by $g(b)$ where $g$ is the indicator function of an intersection of at most $k$ halfspaces in $m$ variables. In the same notation, $f(y)=g(b')$ where $b'$ is obtained from $b$ by flipping the sign of a single random coordinate. Thus, by definition, $\pr(g(b)\neq g(b'))=\frac{1}{m}\as(g)$. Hence,
$$
\ns_\epsilon(f) = \E_g\left[\frac{\as(g)}{m}\right] \leq \frac{O(\sqrt{\log(k)m})}{m} = \sqrt{\frac{\log(k)}{m}} = \sqrt{\epsilon \log(k)}.
$$
This completes our proof.
\end{proof}

Corollary \ref{LearningCor} will now follow by using this bound to bound the weight of the higher degree Fourier coefficients of such an $f$ and then using the $L^1$ polynomial regression algorithm of \cite{learning}.
\begin{proof}[Proof of Corollary \ref{LearningCor}]
Let $f$ be the indicator function of an intersection of $k$ halfspaces. Let $f$ have Fourier transform given by
$$
f = \sum_{S\subset [n]} \chi_S \hat{f}(S).
$$
It is well known that for $\rho\in (0,1)$ that
$$
\ns_\rho(f) = 2\sum_{S\subset [n]} (1-(1-2\rho)^{|S|})|\hat{f}(S)|^2.
$$
Therefore, we have that
$$
\ns_\rho(f) \gg \sum_{|S| > 1/\rho} |\hat{f}(S)|^2.
$$
By Corollary \ref{nsBoundCor}, this tells us that
$$
\sum_{|S| > 1/\rho} |\hat{f}(S)|^2 = O(\sqrt{\rho\log(k)}).
$$
Setting $\rho = \epsilon^2/(C\log(k))$ for sufficiently large values of $C$ yields
$$
\sum_{|S| > C\log(k)\epsilon^{-2}} |\hat{f}(S)|^2 < \epsilon.
$$
Our claim now follows from \cite{learning} Remark 4.
\end{proof}

\section*{Acknowledgements}

This work was done with the support of an NSF postdoctoral research fellowship.

\end{document}